\documentclass[11pt]{article}

\usepackage{fullpage,amsmath,xspace,amssymb,epsfig,stmaryrd,bbm,hyperref,enumerate}%,enumitem}
\usepackage[usenames,dvipsnames]{color}
\usepackage{epstopdf}

\newtheorem{Thm}{Theorem}
\newtheorem{Lem}[Thm]{Lemma}
\newtheorem{Cor}[Thm]{Corollary}

\newtheorem{Def}{Definition}
\newtheorem{Fact}[Thm]{Fact}
\newenvironment{proof}{\noindent {\textbf{Proof }}}{$\Box$ \medskip}

\newcommand\mbR{\mbox{$\mathbb{R}$}}

\newcommand\mbF{\mbox{$\mathbb{F}$}}
\newcommand\B{\{0,1\}}
\newcommand\Bn{\{0,1\}^n}
\newcommand\BntB{\{0,1\}^n\rightarrow \{0,1\}}

\newcommand\pmB{\{+1,-1\}}
\newcommand\pmBn{\{+1,-1\}^n}

\newcommand {\ie} {\textit{i.e.}\xspace}

\newcommand\defeq{\stackrel{\mathrm{\scriptsize def}}{=}}

\newcommand\mob{M\"{o}bius\xspace}

\newcommand\fn[2]{\| \hat{#1} \|_{#2}}

\newcommand\mono{\mbox{\tt mono}\xspace}

\newcommand\alice{\mbox{\sf Alice}\xspace}
\newcommand\bob{\mbox{\sf Bob}\xspace}

\newcommand\gf{\mbox{$\mbF_2$}\xspace}

\newcommand\sen{\mbox{$\sf s$}} % sensitivity
\newcommand\bs{\mbox{$\sf bs$}} % block sensitivity
\newcommand\C{\mbox{$\sf C$}} % certificate sensitivity
\newcommand\Cmin{\mbox{$\sf C_{\min}$}} % minimum certificate sensitivity
\newcommand\Cminc{\mbox{$\sf C_{\min}^{\clo}$}} % minimum certificate sensitivity with closure
\newcommand\Cpminc{\mbox{$\sf C_{\oplus,\min}^{\clo}$}} % minimum parity certificate sensitivity with closure
\newcommand\cover{\mbox{$\sf Cover$}} % cover number
\newcommand\dt{\mbox{\sf {DT}}\xspace} % deterministic decision tree depth
\newcommand\pdt{\mbox{\sf {PDT}}\xspace} % randomized decision tree depth
\newcommand\dcc{\mbox{\sf {CC}}\xspace} % deterministic communication complexity
\newcommand\ncc{\mbox{\sf {N}}\xspace} % deterministic communication complexity
\newcommand\alt{{\sf alt}} % alternating number
\newcommand\dist{\mbox{\sf dist}\xspace}

\newcommand\rank{\mbox{\sf rank}\xspace}

\newcommand\clo{\textup {clo}\xspace}

\newcommand\fand{\mbox{\sf AND}\xspace}
\newcommand\for{\mbox{\sf OR}\xspace}
\newcommand\fxor{\mbox{\sf XOR}\xspace}

\newcommand\fcand{f\circ \wedge}
\newcommand\fcxor{f\circ \oplus}

\newcommand\poly{\textrm{poly}}

\begin{document}
	\title{\bf Sensitivity Conjecture and Log-rank Conjecture for functions with small alternating numbers}
	%\title{\bf Sensitivity Conjecture and Log-rank Conjecture for almost monotone functions}
	\author{Chengyu Lin \qquad Shengyu Zhang}
	%\date{}
	\setcounter{page}{0}
	\maketitle

\begin{abstract}
	The Sensitivity Conjecture and the Log-rank Conjecture are among the most important and challenging problems in concrete complexity. Incidentally, the Sensitivity Conjecture is known to hold for monotone functions, and so is the Log-rank Conjecture for $f(x \wedge y)$ and $f(x\oplus y)$ with monotone functions $f$, where $\wedge$ and $\oplus$ are bit-wise \fand and \fxor, respectively. In this paper, we extend these results to functions $f$ which alternate values for a relatively small number of times on any monotone path from $0^n$ to $1^n$. These deepen our understandings of the two conjectures, and contribute to the recent line of research on functions with small alternating numbers.
\end{abstract}

\newpage
\section{Introduction}
A central topic in Boolean function complexity theory is relations among different combinatorial and computational measures \cite{Juk12}.
%\paragraph{Sensitivity Conjecture}
For Boolean functions, there is a large family of complexity measures such as block sensitivity, certificate complexity, decision tree complex (including its randomized and quantum versions), degree (including its approximate version), etc, that are all polynomially related \cite{BdW02}. One outlier\footnote{There are complexity measures, such as \gf-degree, polynomial threshold degree, total influence, Boolean circuit depth, CNF/DNF size, that are \textit{known} not to belong to the polynomially equivalent class. %, as there are simple examples separating them from the family.
But the position of sensitivity is elusive.} is sensitivity, which {a priori} could be exponentially smaller than the ones in that family. The famous Sensitivity Conjecture raised by Nisan and Szegedy \cite{NS94} says that sensitivity is also polynomially related to the block sensitivity and others in the family. Despite a lot of efforts, the best upper bound we know is still exponential: $\bs(f) \le \C(f) \le \max \left(2^{\sf{s}(f) - 1} \left(\sen(f) - \frac{1}{3}, \sen(f)\right) \right)$ from~\cite{APV15}, improving upon previous work \cite{Sim83,ABGMSZ14}. See a recent survey \cite{HKP11} about this conjecture and how it has resisted many serious attacks. 

%\paragraph{Log-rank Conjecture}

Communication complexity quantifies the minimum amount of communication required for computing functions whose inputs are distributed among two or more parties~\cite{KN97}. In the standard bipartite setting, the function $F$ has two inputs $x$ and $y$, with $x$ given to \alice and $y$ to \bob. The minimum number of bits needed to be exchanged to compute $F(x,y)$ for all inputs $(x,y)$ is the communication complexity $\dcc(F)$. It has long been known \cite{MS82} that the logarithm of the rank of communication matrix $M_F \defeq [F(x,y)]_{x,y}$ is a lower bound of $\dcc(F)$. Perhaps the most prominent and long-standing open question about communication complexity is the Log-rank Conjecture proposed by Lov{\'a}sz and Saks~\cite{LS88}, which asserts that $\dcc(F)$ of any Boolean function $F$ is also upper bounded by a polynomial of $\log \rank(M_F)$. The conjecture has equivalent forms related to chromatic number conjecture in graph theory \cite{LS88}, nonnegative rank \cite{Lov90}, Boolean roots of polynomials over real numbers \cite{Val04}, quantum sampling complexities \cite{ASTS+03,Zha12}, etc. Despite a lot of efforts devoted to the conjecture in the past decades, and the best upper bound is $\dcc(F)=O\big(\sqrt{\rank(M_F)}\log \left(\rank(M_F)\right)\big)$ by Lovett~\cite{Lov14}, which is still exponentially far from the target.

While these two conjectures are notoriously challenging in their full generality, special classes of functions have been investigated. In particular, the Sensitivity Conjecture is confirmed to hold for monotone functions, as the sensitivity coincides with block sensitivity and certificate complexity for those functions \cite{Nis91}. The Log-rank Conjecture is not known to be true for monotone functions, but it holds for monotone functions on two bit-wise compositions between $x$ and $y$. More specifically, two classes of bit-wise composed functions have drawn substantial attention. The first class contains AND functions $F = \fcand$, defined by $F(x,y) = f(x\wedge y)$, where $\wedge$ is the bit-wise \fand of $x,y\in \Bn$. Taking the outer function $f$ to be the $n$-bit \for, we get Disjointness, the function that has had a significant impact on both communication complexity theory itself \cite{She14} and applications to many other areas such as streaming, data structures, circuit complexity, proof complexity, game theory and quantum computation \cite{CP10}. The AND functions also contain %When the outer function $f$ is \parity function, $F$ is the 
other well known functions such as Inner Product, AND-OR trees \cite{JKR09,LS10,JKZ10,GJ15}, and functions exhibiting gaps between communication complexity and log-rank \cite{NW95}. %another well-studied function in bipartite communication complexity. 
The second class is XOR functions $F = \fcxor$, defined by $F(x,y) = f(x\oplus y)$, where $\oplus$ is the bit-wise XOR function. This class includes Equality \cite{Yao79,NS96,Amb96,BK97,BCWdW01} and Hamming Distance \cite{Yao03,GKdW04,HSZZ06,LLZ11,LZ13} as special cases. 

Both AND and XOR functions have recently drawn much attention \cite{LS93,BdW01,ZS09,LZ10,MO10,SW12,LZ13,TWXZ13,Zha14,ODWZST14,Yao15}, partly because their communication matrix rank has intimate connections to the polynomial representations of the outer function $f$. Specifically, the rank of $M_{\fcand}$ is exactly the \mob sparsity\footnote{Named after the \mob transform from $f$ to $\alpha$.} $\mono(f)$, the number of nonzero coefficients $\alpha(S)$ in the multilinear polynomial representation $f(x) = \sum_{S\subseteq[n]} \alpha(S) \prod_{i\in S}x_i$ for $f:\BntB$ \cite{BdW01}. And the rank of $M_{\fcxor}$ is exactly the Fourier sparsity $\fn{f}{0}$, the number of nonzero Fourier coefficients $\hat f(S)$ in the multilinear polynomial representation $f(x) = \sum_{S\subseteq[n]} \hat f(S) \prod_{i\in S}x_i$ for $f:\pmB^n\to\B$. %We will call these two numbers \mob sparsity and Fourier sparsity, respectively, named after the Mobius and Fourier transforms from $f$ to $\alpha$ and to $\hat f$, respectively.

It is known that the Log-rank Conjecture holds for these two classes of functions when the outer function $f$ is monotone \cite{LS93,MO10}, and this work aims to extend these as well as the sensitivity result on monotone functions, to functions that are \textit{close to} being monotone. One needs to be careful about the distance measure here, since the widely-used (e.g. in property testing and computational learning) normalized Hamming distance $\dist(f,g) = \Pr_{x\in \Bn}[f(x) \ne g(x)]$ does not meet our requirement. Indeed, if we flip the value $f(x)$ at just one input $x$, then this changes $f$ by an exponentially small amount measured by $\dist$, but the sensitivity would change from a small $\sen(f)$ to a large $n-\sen(f)$. Similarly, the Fourier sparsity is also very sensitive to local changes ($\fn{f}{0}$ to $2^n - \fn{f}{0}$), and so is \mob sparsity if we flip the value at $0^n$.

%\paragraph{Alternating number}
One robust distance measure to monotone functions, which has recently drawn an increasingly amount of attention, is the alternating number, defined as follows. View the Boolean hypercube $\Bn$ as a lattice with the partial order $x\preceq y$ if $x_i \le y_i$ for all $i$. A path $x^{(1)} \rightarrow \cdots \rightarrow x^{(k)}$ on $\Bn$ is monotone if $x^{(i)} \prec x^{(i+1)}$ for all $i$. %The alternating number of the path is the number of $i$'s with $f(x^{(i)}) \ne f(x^{(i+1)})$. For any function on $\Bn$, the alternating number $\alt(f,x)$ of input $x\in \Bn$ is the maximum alternating number of any monotone path from $0^n$ to $x$, and the alternating number of function $f$ is $\alt(f) = \alt(f,1^n)$. 
The alternating number of a function $f$ on $\Bn$ is the maximum number of $i$'s with $f(x^{(i)}) \ne f(x^{(i+1)})$, on any monotone path $x^{(0)} \rightarrow \cdots \rightarrow x^{(n)}$ from $0^n$ to $1^n$. 
It is clear that constant functions have alternating number 0, and monotone functions have alternating number 1. For general functions $f$, we have $\alt(f) \le n$, thus $\alt(f)$ is a sub-linear complexity measure. %and restricting variables does not increase the alternating number.
The smaller $\alt(f)$ is, the closer it is to monotone functions. Studies of the alternating number dated back to 
\cite{Mar58}, in which Markov showed that the inversion complexity, the minimum number of negation gates needed in any Boolean circuit computing $f$, is exactly $\lceil \log_2(\alt(f) + 1) \rceil$. Late work investigated the inversion complexity/alternating number over computational models such as  constant-depth circuit~\cite{SW93}, bounded-depth circuit~\cite{ST03}, Boolean formula~\cite{Mor09a}, and non-deterministic circuit~\cite{Mor09b}. It is recently showed that small alternating number can be exploited in learning Boolean circuits~\cite{BCOST14}. Also there are some studies in cryptography considering the effect of negation gates~\cite{GMOR15}.

\medskip In this paper, we study the Sensitivity and Log-rank Conjectures for functions whose alternating numbers are small, compared to sensitivity, \mob sparsity and Fourier sparsity. First, the following theorem shows that the Sensitivity Conjecture holds for $f$ with $\alt(f) = \poly(\sen(f))$.
\begin{Thm}\label{thm:alt-sen}
	For any function $f:\BntB$, it holds that \[\bs(f) = O(\alt(f)^2\cdot \sen(f)).\]
\end{Thm}
Note that if a function is non-degenerate in the sense that it depends on all $n$ variables, then the sensitivity is at least $\Omega(\log n)$ \cite{Sim83}, therefore the above theorem also confirms the Sensitivity Conjecture for non-degenerate functions $f$ with $\alt(f) = \poly \log n$. 

The next two theorems confirmed the Log-rank Conjecture for $\fcxor$ with $\alt(f) = \poly\log(\fn{f}{0})$, and for $\fcand$ with $\alt(f) = O(1)$.
\begin{Thm}\label{thm:alt-fs}
	For any function $f:\BntB$, it holds that \[\dcc(\fcxor) \le  2\cdot \alt(f)\cdot \log^2\rank(M_{\fcxor}).\]
\end{Thm}
	
\begin{Thm}\label{thm:alt-mon}
	For any function $f:\BntB$, it holds that \[\dcc(\fcand) \le O(\log^{\alt(f)+1} \rank(M_{\fcand})).\]
\end{Thm}
In the last theorem, the dependence on $\alt(f)$ can be slightly improved (by a factor of 2) if a factor of $\log n$ is tolerated in the communication cost. 
%\paragraph{proof techniques} ...

\medskip
\paragraph{Related work}
The Sensitivity Conjecture has many equivalent forms, summarized in the survey \cite{HKP11}. Also see the recent paper \cite{GKS15} which tries to solve this conjecture using a communication game approach. At the other end of the spectrum, \cite{Rub95,AS11} seek the largest possible separation between sensitivity and block sensitivity. Apart from monotone functions~\cite{Nis91}, the Sensitivity Conjecture has also been confirmed on graph properties~\cite{Tur84}, cyclically-invariant function~\cite{Cha05} and read-once functions~\cite{Mor14}. Other than the conjecture itself, some recent work \cite{AV15,GNSTW15} discussed combinatorial and computational structures of low-sensitivity functions. 

For the Log-rank Conjecture, apart from the equivalent forms mentioned earlier, some seemingly weaker formulations in terms of largest monochromatic rectangle size~\cite{NW95}, randomized communication complexity and information cost~\cite{GL14} are actually equivalent to the original conjecture. For lower bounds, the best one had been $\dcc(F)=\Omega\left((\log\rank(M_F))^{\log_{3} 6}\right)$ (attributed to Kushilevitz in ~\cite{NW95}), achieved by an AND function, until the recent result of $\dcc(F)=\tilde \Omega\left(\log^2\rank(M_F)\right)$ \cite{GPW15}. 
For XOR functions $\fcxor$, the Log-rank Conjecture is confirmed when $f$ is symmetric \cite{ZS09}, monotone \cite{MO10}, linear threshold functions (LTFs)~\cite{MO10}, $AC^0$ functions~\cite{KS13}, has low \gf-degree \cite{TWXZ13} or small spectral norm \cite{TWXZ13}. For AND functions $\fcand$, it seems that the conjecture is only confirmed on monotone functions \cite{LS93}.

\section{Preliminaries}
\paragraph{$n$-bit (Boolean) functions}
We use $[n]$ to denote the set $\{1, 2, \ldots, n\}$. The all-0 $n$-bit string is denoted by $0^n$ and the all-1 $n$-bit string is denoted by $1^n$. 
 
For a Boolean function $f:\BntB$, its \gf-degree is the degree of $f$ viewed as a polynomial over \gf.
Such functions $f$ can be also viewed as polynomials over $\mbR$: $f(x) = \sum_{S\subseteq [n]} \alpha(S) x^S$, where $x^S = \prod_{i\in S} x_i$. If we represent the domain by $\pmBn$, then the polynomial (still over \mbR) changes to $f(x) = \sum_{S\subseteq [n]} \hat f(S) x^S$, usually called Fourier expansion of $f$. The coefficients $\alpha(S)$ and $\hat f(S)$ in the two \mbR-polynomial representations capture many important combinatorial properties of $f$. %The $\ell_p$-norm of the coefficient vector $\alpha\in \mbR^{\Bn}$ is $\|\alpha\|_p = (\sum_{S\subseteq [n]} |\alpha(S)|^p)^{1/p}$ and 
We denote by $\mono(f)$ the \mob sparsity, the number of non-zero coefficients $\alpha(S)$,  %Similarly for $\hat f$, we can define $\|\alpha\|_p = (\sum_{S\subseteq [n]} |\alpha(S)|^p)^{1/p}$. 
and by $\fn{f}{0}$ the Fourier sparsity, the number of non-zero coefficients $\hat f(S)$. 
Some basic facts used in this paper are listed as follows.
\begin{Fact}\label{fact:odd-deg2}
	For any $f:\BntB$, $\deg_2(f) = n$ if and only if $|f^{-1}(1)|$ is odd.
\end{Fact}
\begin{Fact}\label{fact:deg2-fs}
	For any $f:\BntB$, $\deg_2(f) \le \log \fn{f}{0}$.
\end{Fact}

For any input $x \in \Bn$ and $i \in [n]$, let $x^i$ be the input obtained from $x$ by flipping the value of $x_i$. 
For a Boolean function $f:\BntB$ and an input $x$, if $f(x) \ne f(x^i)$, then we say that $x$ is \emph{sensitive to coordinate} $i$, and $i$ is a \emph{sensitive coordinate} of $x$. We can also define these for blocks. For a set $B \subseteq [n]$, let $x^B$ be the input obtained from $x$ by flipping $x_i$ for all $i\in B$. Similarly, if $f(x) \ne f(x^B)$, then we say that $x$ is \emph{sensitive to block} $B$, and $B$ is a \emph{sensitive block} of $x$. The \emph{sensitivity $\sen(f, x)$ of function $f$ on input $x$} is the number of sensitive coordinates $i$ of $x$: $\sen(f, x) = |\{i \in [n] : f(x) \neq f(x^i)\}|$, and the \emph{sensitivity of function} $f$ is $\sen(f) = \max_x \sen(f, x)$. It is easily seen that the $n$-bit \fand and \for functions both have sensitivity $n$. The \emph{block sensitivity $\bs(f, x)$ of function $f$ on input $x$} is the maximal number of disjoint sensitive blocks of $x$, and the \emph{block sensitivity of function} $f$ is $\bs(f) = \max_x \bs(f, x)$. Note that there are always $\bs(f,x)$ many disjoint \textit{minimal} sensitive blocks, in the sense that any $B\subsetneq B_i$ is not a sensitive block of $x$. 

%It is also easily seen that for both \fand and \for function of $n$ variables, 
%\begin{align}\label{eq:AndOr-sen-deg2}
%	\sen(f) = \deg_2(f) = n.
%\end{align}

For a Boolean function $f:\BntB$ and an input $x\in \Bn$, the \emph{certificate complexity $\C(f,x)$ of function $f$ on input $x$} is the minimal number of variables restricting the value of which fixes the function to a constant. The \emph{certificate complexity} of $f$ is $\C(f) = \max_x \C(f,x)$, and the \emph{minimal certificate complexity} of $f$ is $\Cmin(f) = \min_x \C(f,x)$. The \emph{decision tree complexity} $\dt(f)$ of function $f$ is the minimum depth of any decision tree that computes $f$.

A \emph{subfunction} or a \emph{restriction} of a function $f$ on $\Bn$ is obtained from $f$ by restricting the values of some variables. 
%The \emph{dimension} $\dim(f')$ of a subfunction $f'$ is the number of variables not being restricted. And its \emph{co-dimension} $\codim(f')$ is the number of variables being restricted. 
Sometimes we say to restrict $f$ to \emph{above} an input $d$, or to take the subfunction $f'$ over $\{x : x \succeq d\}$, then we mean to restrict variables $x_i$ to be 1 whenever $d_i = 1$. Similarly, we say to restrict $f$ to \emph{under} an input $u$, or take the subfunction $f'$ over $\{x : x \preceq u\}$, meaning to restrict $x_i$ to be 0 whenever $u_i = 0$.
%\red{to check: Is the modification above okay?}
%\begin{Def}
%For a function $f$ on $\Bn$, a partition $(J, \bar{J})$ of $[n]$, and a partial assignment $z \in \B^{\bar{J}}$, define subfunction $f_{J|z} : \B^{J} \rightarrow \B$ to be one obtained from $f$ by fixing the coordinates in $\bar{J}$ to $z$, \ie $f_{J|z} = f(y, z)$ for $y \in \B^{J}$. The dimension of the
%subfunction $f_{J|z}$ is $\dim(f_{J|z}) = |J|$ and the co-dimension is
%$\codim(f_{J|z}) = |\bar{J}|$.
%\end{Def}

%For convenience when we say that to restrict $f$ above input $d$ or take the subfunction $f'$ over $\{x : x \succeq d\}$, it means that $f' = f_{J|z = 1}$ where $J = \{i : d_i = 0\}$. And to restrict $f$ under input $u$ or to take the subfunction $f'$ over $\{x : x\preceq u\}$ means that $f' = f_{J|z = 0}$ where $J = \{i : u_i = 1\}$.

Let $\mathcal F_n$ be the set of all the real-valued functions on $\Bn$. A complexity measure $M:\cup_{n=1}^\infty \mathcal F_n\to \mbR$ is \emph{downward non-increasing} if $M(f') \le M(f)$ for all subfunction $f'$ of $f$. That is, restricting variables does not increase the measure $M$. It is easily seen that the \gf-degree, the alternating number, the decision tree complexity, the sensitivity, the block sensitivity, the certificate complexity, the Fourier sparsity, are all downward non-increasing. When $M$ is not downward non-increasing, it makes sense to define the \emph{closure} by $M^{\clo}(f) = \max_{f'} M(f')$ where the maximum is taken over all subfunctions $f'$ of $f$. In particular, $\Cminc(f) = \max_{f'} \Cmin(f')$. The next theorem relates decision tree complexity to $\Cminc$.

%One can also extend the subcube restriction to affine subspace, leading to the concept of parity certificate complexity. The parity certificate complexity of a function $f$ on an input $x\in \Bn$ is
%\[
%\C_\oplus(f,x) = \min\{\codim(H): x\in H, \text{$H$ is an affine subspace, on which $f$ is constant}\}.
%\]
%The parity certificate complexity of $f$ is $\C_\oplus (f) = \max_x \C_\oplus(f,x)$ and the minimum parity certificate complexity is $\Cpmin(f) = \min_x \C_\oplus(f,x)$. We can define its closure variant $\Cpminc(f) = \max_{f': subfunction} \Cpmin(f')$, where the maximization is over all subfunctions $f'$ obtained from restricting $f$ on some affine subspace $H$ of $\Bn$. The following theorem \cite{TWXZ13} reduces the task of designing efficient protocols for $\fcxor$ to that of upper bounding $\Cpmin(f)$.
%\begin{Thm}[\cite{TWXZ13}]\label{thm:pdt-deg-reduction}
%	For any $f:\BntB$, it holds that $\pdt(f) \le \Cpminc(f)\deg_2(f)$.
%\end{Thm}
\begin{Thm}[\cite{TWXZ13}]\label{thm:dt-deg-reduction}
	For any $f:\BntB$, it holds that $\dt(f) \le \Cminc(f)\deg_2(f)$.
\end{Thm}
(The original theorem proved was actually $\pdt(f) \le \Cpminc(f)\deg_2(f)$, where $\pdt(f)$ is the parity decision tree complexity and $\Cpminc(f)$ is the parity minimum certificate complexity. But as observed by \cite{Tsa14}, the same argument applies to standard decision tree as well.)

\medskip For general Boolean functions $f$, we have $\sen(f) \le \bs(f) \le \C(f)$. But when $f$ is monotone, equalities are achieved.
\begin{Fact}\label{fact:s-bs-C monotone}
	If $f:\BntB$ is monotone, then $\sen(f) = \bs(f) = \C(f)$.
\end{Fact}
\begin{Fact}[\cite{MO10}]\label{fact:sen-deg2}
	If $f:\BntB$ is monotone, then $\sen(f) \le \deg_2(f)$.
\end{Fact}

One can associate a partial order $\preceq$ to the Boolean hypercube $\Bn$: $x\preceq y$ if $x_i \le y_i$ for all $i$. We also write $y \succeq x$ when $x\preceq y$. If $x\preceq y$ but $x\ne y$, then we write $x \prec y$ and $y\succ x$. A path $x^{(1)} \rightarrow \cdots \rightarrow x^{(k)}$ on $\Bn$ is \textit{monotone} if $x^{(i)} \prec x^{(i+1)}$ for all $i$.
\begin{Def}
	For any function on $\Bn$, the \emph{alternating number of a path} $x^{(1)} \rightarrow \cdots \rightarrow x^{(k)}$ is the number of $i\in \{1,2,...,k-1\}$ with $f(x^{(i)}) \ne f(x^{(i+1)})$. The alternating number $\alt(f,x)$ of input $x\in \Bn$ is the maximum alternating number of any monotone path from $0^n$ to $x$, and the \emph{alternating number of a function} $f$ is $\alt(f) = \alt(f,1^n)$. Equivalently, one can also define $\alt(f)$ to be the largest $k$ such that there exists a list $\{x^{(1)}, x^{(2)}, \ldots, x^{(k+1)}\}$ with $x^{(i)} \preceq x^{(i+1)}$ and $f(x^{(i)}) \neq f(x^{(i+1)})$, for all $i\in [k]$. 
%\red{Chengyu: Check whether this is ok. This slightly differs from your definition since yours didn't specify what ``length'' means. I also changed the $x_i$ to $x^{(i)}$ to avoid confusion with the $i$-th variable of $x$. I dropped the ``Boolean'' condition.}
\end{Def}
A function $f:\Bn\to\mbR$ is \emph{monotone} if $f(x) \le f(y)$, $\forall x\preceq y$. A function $f:\Bn\to\mbR$ is \emph{anti-monotone} if $f(x) \le f(y)$, $\forall x\succeq y$. It is not hard to see that $\alt(f) = 0$ iff $f$ is constant, and $\alt(f) = 1$ iff $f$ is monotone or anti-monotone.

%\begin{Def}
%For boolean function $f : \Bn \rightarrow \B$, the alternating number $\alt(f)$ is the length of the longest list $\{x_i\}$ such that $x_{i-1} \preceq x_i$ and $f(x_{i-1}) \neq f(x_i)$.
%\end{Def}

\begin{Def}
For a function $f$ on $\Bn$, an input $u \in \Bn-\{1^n\}$ is called a \emph{max term}  if
$f(u) \neq f(1^n)$, and $f(x) = f(1^n)$ for all $x \succ u$. An input $d \in \Bn-\{0^n\}$ is called a \emph{min term} if $f(d) \neq f(0^n)$, and $f(x) = f(0^n)$ for all $x \prec d$.
%\red{Chengyu: I combined the two definitions, changed $\succeq$ to $\succ$ and $\preceq$ to $\prec$, and dropped the Boolean condition}.
\end{Def}

%For $x \in \Bn$, let $\wth(x)$ be the Hamming weight of $x$ where $\wth(x) = \left|\{i : x_i = 1\}\right|$.

%\begin{Fact}
%For any max term $u$ of function $f$, ${\sf wt_H}(u) \ge n - \sen(f)$.
%For any min term $d$ of function $f$, ${\sf wt_H}(d) \le \sen(f)$.
%\end{Fact}

\paragraph{Communication complexity}
Suppose that for a bivariate function $F(x,y)$, the input $x$ is given to \alice and $y$ to \bob. The (deterministic) \emph{communication complexity} $\dcc(F)$ is the minimum number of bits needed to be exchanged by the best (deterministic) protocol that computes $F$ (on the worst-case input). 

The rank (over \mbR) of the communication matrix for bit-wise composed functions coincides with some natural parameters of the outer function $f$. For XOR functions $\fcxor$, it holds that $\rank(M_{\fcxor}) = \fn{f}{0}$, and for AND functions $\fcand$, it holds that $\rank(M_{\fcand}) = \mono(f)$. When $f$ is \for function of $n$ variables, we have $\rank(M_{\fcand}) = \mono(\for_n) = 2^n-1$.

It is well known that communication can simulate queries. More specifically, for XOR functions and AND functions, we have that \begin{align}\label{eq:cc-dt}
	\dcc(\fcand) \le 2\dt(f) \text{ and } \dcc(\fcxor) \le 2\dt(f).
\end{align}

In a \B-communication matrix $M$, a \emph{1-rectangle} is a all-1 submatrix. The \emph{1-covering number} $\cover_1(M)$ of matrix $M$ is the minimum number of 1-rectangles that can cover all 1 entries in $M$. (These 1-rectangles need not be disjoint.) For notational convenience, we sometimes write $\cover_1(F)$ for $\cover_1(M_F)$. Lov{\'{a}}sz \cite{Lov90} showed the following upper bound.
\begin{Thm}[\cite{Lov90}]\label{thm:CC-N-rank}
	For any Boolean funcion $F(x,y)$, it holds that $\dcc(F) \le \log\cover_1(M_F)\cdot \log\rank(M_F)$.
\end{Thm}

\section{The Sensitivity Conjecture}

%\begin{Lem}\label{lemma:certificate}
%For any Boolean function $f$, $\C(f, 0^n) \le \alt(f) \cdot \sen(f)$ and $\C(f, 1^n) \le \alt(f) \cdot \sen(f)$.
%\end{Lem}
This section is devoted to the proof of Theorem \ref{thm:alt-sen}. We will first show the following lemma, in which the first statement is used in this section and the second statement will be used in Section \ref{sec:logrank} for proving the Log-rank Conjecture of XOR functions.
\begin{Lem}\label{thm:Cmin-alt}
	For any $f:\BntB$, it holds that
	\begin{enumerate}
		\item $\max\{\C(f,0^n),\C(f,1^n)\} \le \alt(f)\cdot \sen(f)$
		\item $\max\{\C(f,0^n),\C(f,1^n)\} \le \alt(f)\cdot \deg_2(f)$.
	\end{enumerate}
\end{Lem}
\begin{proof}
	First note that it suffices to prove the two upper bounds for $\C(f,0^n)$, because then we can take $g(x) = f(\bar x)$ to get that $\C(f,1^n) = \C(g,0^n) \le \alt(g)\cdot \sen(g) = \alt(f)\cdot \sen(f)$.
	
	We prove upper bounds on $\C(f,0^n)$ by induction on $\alt(f)$. When $\alt(f) = 1$, the function is either monotone or anti-monotone, thus \[\C(f,0^n) \le \C(f) = \sen(f) \le \deg_2(f),\]
	where the first inequality is by definition of $\C(f,0^n)$, the middle equality is by Fact \ref{fact:s-bs-C monotone} and the last inequality is because $\sen(f) \le \deg_2(f)$ for monotone $f$ (Fact \ref{fact:sen-deg2}). Now we assume that the inequalities in the lemma hold for $\alt(f) < a$ and we will show that they hold for $f$ with $\alt(f) = a$ as well. 
	%Take a monotone path $x^{(0)}\rightarrow \cdots \rightarrow x^{(n)}$ from $0^n$ to $1^n$ with $a = \alt(f)$ alternations. Consider the last alternating point $x = x^{(j)}$, namely the last point with $f(x^{(j)})\ne f(x^{(j+1)})$. Trace from $x$ upwards with function value unchanged to reach $y$. That is, we walk from $x$ along any non-sensitive edge towards $1^n$, until no such edge exists, and call that point $y$. 
	Let $u$ be a max term of $f$.
	Define $S_0(u) \defeq \{i\in [n]: u_i = 0\}$, and consider the subcube above $u$: $\{x: x \succeq u\}$. Let $f'$ be the subfunction obtained by restricting $f$ on this subcube. By the definition of max term $f(u) \ne f(u^i)$ for all $i\in S_0(u)$. Therefore,
	\begin{align}\label{eq:S0sen}
	|S_0(u)| \le \sen(f,u) \le \sen(f).
	\end{align}
	%Recall that the path $x^{(0)}\rightarrow \cdots \rightarrow x^{(n)}$ has $a=\alt(f)$ alternations, thus $\alt(f,x^{(n)}) = a$, which also implies that $\alt(f,y) = a$. Since $a$ is the maximum number of alternations on any monotone path from $0^n$ to $1^n$, 
	We know that any point $z \succ u$ has $f(z) = f(1^n) \ne f(u)$. So the number of 1-inputs of $f'$ %(or equivalently, the number of input $t$ of $f$ in the subcube $\{t\in \Bn: t\succeq u\}$ where  $f(t) = 1$) 
	is odd, implying that $\deg_2(f') = |S_0(u)|$ (Fact \ref{fact:odd-deg2}). Thus we have
	\begin{align}\label{eq:S0deg}
	|S_0(u)| = \deg_2(f') \le \deg_2(f).
	\end{align}

	Now consider another restriction of $f$, this time to the subcube \emph{under} $u$, \ie $\{x: x\preceq u\}$. This is implemented by restricting all variables in $S_0(u)$ to $0$, yielding a subfunction $f''$ with $\alt(f'') \le \alt(f) - 1$. Using induction hypothesis, we have that
	\begin{align}\label{eq:Cminsen}
	\C(f'',0^{[n]-S_0(u)}) \le \alt(f'') \cdot \min\{\sen(f''),\deg_2(f'')\} \le (\alt(f) - 1) \cdot \min\{\sen(f),\deg_2(f)\} 
	\end{align} %and that
	%\begin{align}\label{eq:Cmindeg}
	%\C(f'',0^{[n]-S_0(u)}) \le \alt(f'')  \cdot \deg_2(f'') \le (\alt(f) - 1) \cdot \deg_2(f).
	%\end{align}
	Recall that $f''$ is obtained from $f$ by restricting $|S_0(u)|$ variables, thus \[\C(f,0^n) \le |S_0(u)| + \C(f'',0^{[n]-S_0(u)}).\]
	Plugging Eq.\eqref{eq:S0sen} and Eq.\eqref{eq:Cminsen} into the above inequality gives
	\[\C(f,0^n) \le \alt(f) \cdot \min\{\sen(f),\deg_2(f)\},\]
	%And plugging Eq.\eqref{eq:S0deg} and Eq.\eqref{eq:Cmindeg} gives 
	%\[\C(f,0^n) \le \alt(f) \cdot \deg_2(f)\]
	completing the induction.
\end{proof}

Now we are ready to prove the following theorem, which gives an explicit constant for Theorem \ref{thm:alt-sen}.
\begin{Thm}\label{thm:general}
For any boolean function $f$, %$\bs(f) \le O\left(\alt(f)^2\right) \cdot \sen(f)$.
\begin{equation}\label{eqn:general}
\bs(f) \le
\begin{cases}
C_t \cdot \sen(f) & \text{ if }\alt(f) = 2t, \\
(C_t + 1) \cdot \sen(f) & \text{ if }\alt(f) = 2t+1,
\end{cases}
\end{equation}
where $C_t = \sum_{i=1}^t (i+2) = \frac{1}{2}t(t+5)$.
\end{Thm}

\begin{proof}
We prove Eq.\eqref{eqn:general} by induction on $t = \lfloor \alt(f)/2\rfloor$. Clearly it holds when $t=0$: If $\alt(f)=0$ then $f$ is a constant function and $\bs(f) = \sen(f) = 0$. When $\alt(f) = 1$, $f$ is monotone or anti-monotone, thus $\bs(f) = \sen(f)$.

Now for any Boolean function $f$ with $\alt(f) > 1$, we first consider the case when $\alt(f) = 2t \ge 2$. %WLOG let's assume that $f(0^n) = f(1^n) = 0$.
We will bound the block sensitivity for each input $x$.
Consider the following possible properties for $x$.

\begin{enumerate} %[(1)]
	\item\label{cond:maxterm} there exists a max term $u$ of $f$ such that $x \preceq u$;
	\item\label{cond:minterm} there exists a min term $d$ of $f$ such that $x \succeq d$.
\end{enumerate}

{\bf Case 1}: $x$ satisfies at least one of the above conditions.
Without loss of generality assume it satisfies the first one; {the other case can be similarly argued}. %the previous one holds, otherwise we can consider $g(x) = f(\neg x)$ .
Fix such a max term $u\succeq x$.
By definition of max term, we know that $\alt(f,u) \le \alt(f) - 1$, and that $u$ is sensitive to all $i\in S_0(u)\defeq \{i: u_i = 0\}$. Therefore, $|S_0(u)| \le \sen(f,u) \le \sen(f)$.

Let $f'$ be the subfunction of $f$ restricted on the subcube $\{t : t \preceq u\}$, then $\alt(f') = \alt(f,u) \le \alt(f) - 1 = 2t-1 = 2(t-1)+1$.
%We consider the subfunction of $f$ over all $x \preceq u$,
%\ie let $J = \{i \in [n]: u_i = 1\}$ and consider $f_{J|z=0}$.
By induction hypothesis and the fact that sensitivity is downward non-increasing, we have
\begin{align}\label{eq:bs-ind}
	\bs(f', x) \le \bs(f') \le (C_{t-1} + 1) \cdot \sen(f')
\le (C_{t-1} + 1) \cdot \sen(f).
\end{align}
%Clearly taking a subfunction does not increase the sensitivity so $\sen(f') \le \sen(f)$. Moreover because $\alt(f') \le 2t - 1$, so for any $x \preceq u$, by

{Next it is not hard to see that
\begin{align}\label{eq:bs-res}
	\bs(f, x) \le \bs(f', x) + |S_0(u)|.	
\end{align}
Indeed, take any disjoint minimal sensitive blocks $B_1, \ldots, B_\ell\subseteq [n]$ of $x$ (with respect to $f$), where $\ell = \bs(f,x)$. If $B_i \subseteq [n]-S_0(u)$, then $x$ is still sensitive to $B_i$ in $f'$. As the $B_i$'s are disjoint, at most $|S_0(u)|$ many $B_i$'s are not contained in $[n]-S_0(u)$, thus at least $\bs(f,x) - |S_0(u)|$ blocks $B_i$ are still sensitive blocks of $x$ in $f'$. Therefore, $\bs(f,x) - |S_0(u)| \le \bs(f',x)$, as Eq.\eqref{eq:bs-res} claimed.}

Combining Eq.\eqref{eq:bs-ind}, Eq.\eqref{eq:bs-res}, and the fact that $|S_0(u)| \le \sen(f)$, we conclude that
\begin{align}\label{eq:case1bound}
	\bs(f, x)  \le \bs(f', x) + |S_0(u)| \le (C_{t-1} + 1) \cdot \sen(f') + \sen(f)
\le (C_{t-1} + 2) \cdot \sen(f),
\end{align}
which is at most $C_t \cdot \sen(f)$ by our setting of of parameter $C_t = \sum_{i=1}^{t}(i+2) = C_{t-1} + t + 2$.

\medskip
{\bf Case 2:} $x$ satisfies neither of the conditions~\ref{cond:maxterm} and~\ref{cond:minterm}.
So $f(x)$ needs to be the same with both $f(0^n)$ and $f(1^n)$, and $f$ is constant on both subcubes $\{t: t\succeq x\}$ and $\{t: t\preceq x\}$. Otherwise we can take a minimal $d$ where $d \preceq x$ and
$f(d) = f(x) \neq f(0^n)$ and by definition $d$ is a min term, 
or take the maximal $u$ where $u \succeq x$ and
$f(u) = f(x) \neq f(1^n)$ and by definition $u$ is a max term.

Fix $\ell = \bs(f, x)$ disjoint minimal sensitive blocks $\{B_1, B_2, \dots, B_\ell \}$ of $x$. For each block $B_i$, decompose it into
$B_i = U_i \cup D_i$ where $U_i = \{i \in B_i : x_i = 1\}$ and
$D_i = \{i \in B_i : x_i = 0\}$, as depicted below.
\[
x = (\underbrace{\overbrace{0\dots 0}^{D_1}\overbrace{1\dots 1}^{U_1}}_{B_1})(\underbrace{\overbrace{0\dots 0}^{D_2}\overbrace{1\dots 1}^{U_2}}_{B_2}) \cdots (\underbrace{\overbrace{0\dots 0}^{D_l}\overbrace{1\dots 1}^{U_l}}_{B_l}) 0 \dots 01 \dots 1
\]

First we will show that for each $i$, $x^{U_i}$ satisfies condition~\ref{cond:maxterm}
and $x^{D_i}$ satisfies condition~\ref{cond:minterm},
\ie there exist some max term $u\succeq x^{U_i}$ and some min term $d\preceq x^{D_i}$.
(See figure~\ref{fig:block} for an illustration.)
Indeed, for any sensitive block $B_i$ of $x$, $f(x^{B_i}) \neq f(x) = f(0^n) = f(1^n)$.
Take a \textit{maximal} $u_i$ such that $u_i \succeq x^{B_i}$ and $f(u_i) = f(x^{B_i})$.
By definition $u_i$ is a max term. Similarly we can take a min term $d_i$ where $d_i \preceq x^{B_i}$.
Then from the definition of $U_i$ and $D_i$ we can conclude that
$x^{U_i} \preceq x^{B_i} \preceq u_i$ and $x^{D_i} \succeq x^{B_i} \succeq d_i$.
Moreover, both $U_i$ and $D_i$ cannot be empty, since otherwise either
$x \preceq x^{D_i} = x^{B_i} \preceq u_i$ or $x \succeq x^{U_i} = x^{B_i} \succeq d_i$, contradicting our assumption of \textbf{case 2}. This further indicates that
$f(x) = f(x^{U_i}) = f(x^{D_i})$ as we have taken each $B_i$ to be a \textit{minimal} sensitive block.

\begin{figure}[ht]
    \centering
    \includegraphics{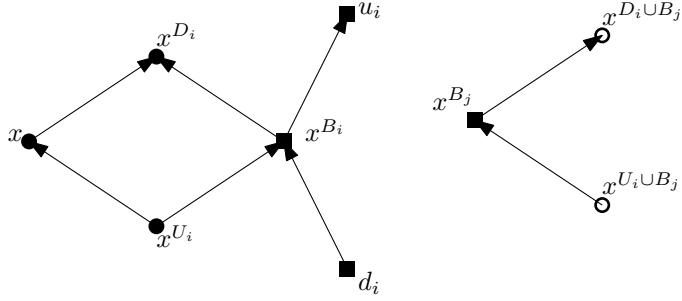}
    \caption{\label{fig:block}Order among different inputs used in the proof. Arrows indicate the partial order in $\Bn$. Solid round circles stand for one Boolean value, and squares stand for the other. The value for hollow circles are not fully determined, but we will show that most of them share the same value with the squares.}
\end{figure}

Next we are going to find some $U_i$ or $D_i$ such that
$x^{U_i}$ or $x^{D_i}$ is sensitive to most $B_i$'s. 
In this case if there are many sensitive blocks of input $x$, $x^{U_i}$ or $x^{D_i}$ must have high block sensitivity.
But we have eliminate this possibility in case 1.
%(Clearly block $D_i$ is sensitive to $x^{U_i}$ and block $U_i$ is sensitive to $x^{D_i}$.)
To achieve this, we count the following two quantities:
\begin{itemize}
\item $\#U$ : the number of pairs $(i, j)$ such that $i \neq j$ and $f(x^{U_i}) \neq f(x^{U_i \cup B_j})$
\item $\#D$ : the number of pairs $(i, j)$ such that $i \neq j$ and $f(x^{D_i}) \neq f(x^{D_i \cup B_j})$
\end{itemize}

Recall that $f(x) = f(x^{U_i}) = f(x^{D_i})$ and $f(x) \neq f(x^{B_j})$, thus it is equivalent to counting
\begin{itemize}
\item $\#U$ : the number of pairs $(i, j)$ such that $i \neq j$ and $f(x^{B_j}) = f(x^{U_i \cup B_j})$
\item $\#D$ : the number of pairs $(i, j)$ such that $i \neq j$ and $f(x^{B_j}) = f(x^{D_i \cup B_j})$
\end{itemize}

Now we bound the number of such $i$'s for each $j$. Fix a block $B_j$, and
consider the subfunction $f^u$ on the subcube $\{z : z \succeq x^{B_j}\}$
and the subfunction $f^d$ on the subcube $\{z : z \preceq x^{B_j}\}$. Let us look at $f^u$ first. Because $D_i \cap B_j = \emptyset$ whenever $i\neq j$,
$x^{D_i \cup B_j} \succeq x^{B_j}$ which lies in the domain of $f^u$.
By the definition of certificate complexity of $f^u$ on input
$x^{B_j}$, there is a subcube $C$ of co-dimension $\C(f^u, x^{B_j})$ (with respect to $\{z : z \succeq x^{B_j}\}$) containing $x^{B_j}$, s.t. $f$ takes a constant 0/1 value on $C$. Denote by $S$ the set of coordinates in this certificate. Then $S\subseteq \{k\in [n]: (x^{B_j})_k = 0\}$ and $|S| = \C(f^u, x^{B_j})$. Now for each $D_i$, if $D_i\cap S = \emptyset$, then $f(x^{B_j}) = f(x^{D_i \cup B_j})$ as the values of the certificate variables $S$ are not flipped. As all $\{D_i\}_{i\neq j}$ are disjoint, at most $\C(f^u, x^{B_j})$ many of $D_i$'s may intersect $S$. Thus 
$f(x^{B_j}) = f(x^{D_i \cup B_j})$ for all but at most $\C(f^u, x^{B_j})$ many of $D_i$.
Similarly we can say that all
but at most $\C(f^d, x^{B_j})$ many of $U_i$'s $(i\neq j)$ satisfy that
$f(x^{B_j}) = f(x^{U_i \cup B_j})$.
Applying Lemma~\ref{thm:Cmin-alt} (statement 1), we have
\[
\C(f^u, x^{B_j}) \le \alt(f^u) \cdot \sen(f^u) \le \alt(f^u) \cdot \sen(f),
\]
\[
\C(f^d, x^{B_j}) \le \alt(f^d) \cdot \sen(f^d) \le \alt(f^d) \cdot \sen(f).
\]
Because $\alt(f^u) + \alt(f^d) \le \alt(f) = 2t$, and there are $\ell$ sensitive blocks $B_i$, thus from the second definition of $\#U$ and $\#D$ we can see that
\begin{align}\label{eq:ZeroInMatrix}
\#U + \#D \ge \ell \cdot \left((\ell - 1 - \alt(f^u) \cdot \sen(f)) + (\ell - 1 - \alt(f^d) \cdot \sen(f)) \right)
\ge \ell \cdot 2\left(\ell - 1 - t \cdot \sen(f)\right).
\end{align}

%Then we use this to bound the number of $j$'s for some $i$.
Consider a $2\ell\times \ell$ matrix $M$ of $\B$-entries as follows. The rows are indexed by $U_i$ and $D_i$, and the columns are indexed by $B_j$. For each entry $(T_i, B_j)$ where $T_i$ is $U_i$ or $D_i$, if $i=j$ then let the entry be 1; if $i\ne j$, then let it be 1 when $f(x^{B_j}) \ne f(x^{T_i \cup B_j})$ and 0 otherwise. Note that $\#U + \#D$ is exactly the number of zeros in the matrix $M$, thus the inequality Eq.\eqref{eq:ZeroInMatrix} says that the number of 1's in the matrix is at most $2\ell + 2\ell \cdot t \cdot \sen(f)$.
Since the total number of $\{U_i\}$ and $\{D_i\}$ is $2\ell$, on average each row has at most $t\cdot \sen(f) + 1$ ones. Thus there exists some row $T_i$ ($T_i$ being either $U_i$ or $D_i$) with at most $t\cdot \sen(f)$ ones on columns $B_j$ with $j\ne i$. For this row $T_i$, the number of $j$'s such that $i \neq j$ and $f(x^{T_i}) = f(x) \neq f(x^{B_j}) = f(x^{T_i \cup B_j})$ is no smaller than $\ell - 1 - t \cdot \sen(f)$. Considering that $x^{T_i}$ is also sensitive to $B_i\backslash T_i$, we conclude that
\[
\bs(f, x^{T_i}) \ge 1 + (\ell -1 - t \cdot \sen(f)) = \bs(f, x) - t \cdot \sen(f).
\]
Finally, recall that we have showed that $x^{T_i}$ satisfies one of the condition~\ref{cond:maxterm} and~\ref{cond:minterm}. Therefore $x^{T_i}$ is an input falling into \textbf{case 1}. By Eq.\eqref{eq:case1bound}, we have $\bs(f, x^{T_i}) \le (C_{t-1} + 2) \cdot \sen(f)$. Putting everything together, we have %Recall that $C_t = \sum_{i=1}^t (i + 2) = C_{t-1} + (t + 2)$
\[
	\bs(f, x) \le \bs(f, x^{T_i}) + t \cdot \sen(f) \le (C_{t-1} + 2 + t) \cdot \sen(f) = C_t \cdot \sen(f).
\]
This finishes the proof for $\alt(f) = 2t$.

When $\alt(f) = 2t + 1$. For any input $x$, $f(x)$ must differ from either
 $f(0^n)$ or $f(1^n)$ since $f(0^n) \neq f(1^n)$. Without loss of generality, assume that
$f(x) \neq f(0^n)$. Take the minimal $d$ such that $d \preceq x$ and
$f(d) = f(x) \neq f(0^n)$. By definition $d$ is a min term and
$x$ satisfies condition~\ref{cond:minterm}.
Then using the same analysis above as in \textbf{case 1}, we can show
$\bs(f, x) \le (C_t + 1) \cdot \sen(f)$ and this finishes the proof.
\end{proof}

\section{The Logrank Conjecture}\label{sec:logrank}
We prove Theorem \ref{thm:alt-fs} and \ref{thm:alt-mon} in this section. %, with the former in Section \ref{sec:XOR} and the latter in Section \ref{sec:AND}.
%\subsection{XOR functions}\label{sec:XOR}
We start with Theorem \ref{thm:alt-fs}, which is now easy given Lemma \ref{thm:Cmin-alt}. Recall that the second statement of Lemma \ref{thm:Cmin-alt} says that $\max\{\C(f,0^n),\C(f,1^n)\} \le \alt(f)\cdot \deg_2(f)$, therefore 
\begin{align}\label{eq:Cmin-alt-deg2}
	\Cmin(f) \le \alt(f)\cdot \deg_2(f).
\end{align} 
As both $\alt(f)$ and $\deg_2(f)$ are downward non-increasing, applying Eq.\eqref{eq:Cmin-alt-deg2} to all subfunctions of $f$ yields $\Cminc(f)\le \alt(f)\cdot \deg_2(f)$. Since $\dt(f) \le \Cminc(f) \cdot \deg_2(f)$ (Theorem \ref{thm:dt-deg-reduction}) we get the following.
\begin{Thm}\label{thm:DT-alt}
	For any $f:\BntB$, it holds that $\dt(f) \le \alt(f) \cdot \deg_2(f)^2$.
\end{Thm}
Theorem \ref{thm:alt-fs} follows from this together with the fact that  $\dcc(\fcxor) \le 2\dt(f)$ (Eq.\eqref{eq:cc-dt}) and that $\deg_2(f) \le \log\fn{f}{0} = \log\rank(M_{\fcxor})$ (Fact \ref{fact:deg2-fs}). %$\dcc(f\circ \oplus) \le 2\alt(f) \cdot \log^2\rank(M_{f\circ \oplus})$.

Note that if we use the first statement of Lemma \ref{thm:Cmin-alt}, we will get the following corollary, which gives better dependence on $\alt(f)$ for low \gf-degree functions.
\begin{Cor}
	$\dt(f) \le \alt(f)\sen(f) \cdot \deg_2(f)$.
\end{Cor}

%\subsection{AND functions}\label{sec:AND}
%For any $f:\BntB$, an input $x$ is a \textit{max term} if $f(x) \ne f(y)$ for all $y \succeq x$, and an input $x$ is a \textit{min term} if $f(x) \ne f(y)$ for all $y \preceq x$. For an input $x\in \Bn$, let $H(x)$ be its Hamming weight.
\medskip Next we prove Theorem \ref{thm:alt-mon} for AND functions. Different than the above approach for XOR functions of going through $\dt(f)$, we directly argue communication complexity of AND functions. 
Recall that Theorem \ref{thm:alt-mon} says that
	\[\dcc(\fcand) \le \min\{O(\log^{a+1} \rank(M_{\fcand})), O(\log^{\frac{a+3}{2}} \rank(M_{\fcand}) \log n)\}.\]
\begin{proof}(of Theorem \ref{thm:alt-mon})
	Without loss of generality, we can assume that $f(0^n) = 0$ since otherwise we can compute $\neg f$ first and negate the answer (note that $\rank(M_{\neg f \circ \wedge})$ differs from $\rank(M_{\fcand})$ by at most 1). For notational convenience let us define $r = \mono(f) = \rank(M_{\fcand})$ and $\ell = \log r$. For $b\in \B$, further define $C_b^{(a)}$ to be the maximum $\cover_b(\fcand)$ over all functions $f:\BntB$ with alternating number $a$ and $f(0^n) = 0$. We will give three bounds for $C_b^{(a)}$ in terms of $C_b^{(a-1)}$, and combining them gives the claimed result in Theorem \ref{thm:alt-mon}.
	
	\medskip
	\textbf{Bound 1, from max terms.} We apply this bound for $C_b^{(a)}$ when $a$ and $b$ have different parities, that is, when $a$ is even and $b = 1$, and when $a$ is odd and $b = 0$. Consider the first case and the second is similar. Take any Boolean function $f$ with $f(0^n) = 0$ and $\alt(f) = a$ is even, we have $f(1^n) = 0$. Any 1-input is under some max term, so it is enough to cover inputs under max terms when bounding the $\cover_1(f)$. Take an arbitrary max term $u\in \Bn$. Suppose its Hamming weight is $s$. Considering the subfunction $f'$ on $\{t:t\succeq u\}$, which is %Restricting variables $x_i$ for all $\{i: u_i = 1\}$ to 1, %and let $s$ be the Hamming weight of $u$ where $s = |\{i : u_i = 1\}|$, we get 
	an OR function of $n-s$ variables. In the communication setting, this is the Disjointness function of $n-s$ variables. Thus $\ell = \log \rank(M_{\fcand}) \ge n-s$. This implies that all max terms $u$ of $f$ are $\ell$-close to $1^n$ in Hamming distance. Considering that different max terms are incomparable by definition, we know that the number of max terms is at most $\binom{n}{\ell}$. 
	
	Next we upper bound the 1-rectangles by giving a partition of set of 1-inputs into 1-rectangles. For each max term $u\in \Bn$, let $U = \{i\in [n]: u_i = 1\}$, and $k = n-|U|$, then $k \le \ell$. The submatrix $\{(x,y): x, y\in \Bn, x\wedge y \preceq u\}$ is partitioned into $3^k$ submatrices as follows. Suppose that the set of 0-coordinates in $u$ is $\{i_1, \ldots, i_k\}$, then for each $i_j$, we can choose  $(x_{i_j},y_{i_j})$ from the set $\{(0,0), (0,1), (1,0)\}$ to enforce $x_{i_j}\wedge y_{i_j} = 0$. Thus there are $3^k$ ways of restricting these $k$ variables in $\bar U$, giving $3^k$ submatrices.  Let $f_u:\B^U \to \B$ be the subfunction of $f$ restricted on the subcube $\{t : t \preceq u\}$ where $f_u(z_U) = f(z_U,0_{\bar U})$. (Here the input to $f$ is $x'\wedge y'$ at $U$ and 0 at $\bar U$.) Note that each of the $3^k$ submatrices is still the communication matrix of $f_u\circ \wedge$ for some max term $u$. %an AND function, \ie $f_u\circ \wedge$. %$F': \B^U \times \B^U \rightarrow \B$, defined by $F'(x',y') = f(x'\wedge y', 0_{\bar U})$.
	Also note that this $f_u$ has $f_u(0_U) = 0$, but $f_u(1_U) = 1$ and $\alt(f_u) \le \alt(f) - 1$. Since all the 1-inputs of $f$ are under some max term $u$, the 1-covering number $\cover^1(\fcand)$ can be upper bounded by the following:
	\[\cover_1(\fcand) \le \sum_{u: \text{max term}} 3^\ell \cdot \cover_1(f_u\circ \wedge) \le \binom{n}{\ell} \cdot 3^\ell \cdot  \max_{u: \text{max term}}\cover_1(f_u\circ \wedge).\] 
	Using the fact $\alt(f_u) \le \alt(f) - 1$, and that the above inequality holds for any $f$, we have the following bound on $C_1^{(a)}$: 
	\begin{align}\label{eq:C1-even}
		\log C_1^{(a)} \le 3\ell \cdot \log n + \log C_1^{(a-1)}, \text{ when $a$ is even}.
	\end{align}
	Similarly, when $a$ is odd, $f(1^n) = 1$, and thus any 0-input is under some max term. A similar argument shows the following bound on $C_0^{(a)}$:
	\begin{align}\label{eq:C0-odd}
		\log C_0^{(a)} \le 3\ell \cdot \log n + \log C_0^{(a-1)}, \text{ when $a$ is odd}.
	\end{align}

	\medskip
	\textbf{Bound 2, from min terms.} Take any Boolean function $f$ with $f(0^n) = 0$. Then any 1-input must be above some min term. Take any min term $d$. Let $D = \{i:d_i = 1\}$. If we restrict variables $x_i$ and $y_i$ to 1 for all $i\in D$, then we go to a rectangle $\{(x,y): x_i = y_i = 1, \forall i\in D\}$. The union of these rectangles for all min terms $d$ contains all 1-inputs. Restrict $f$ on the subcube $\{z:z\succeq d\}$ to get a subfunction $f_d$, which has $f_d(0^{\bar D}) = 1$, and $\alt(f_d) \le \alt(f) - 1$. 
	Note that for each min term $d$, we have $\alpha(d) = \sum_{x\preceq d} (-1)^{|d\oplus x|} f(x) = 1 \ne 0$ \footnote{If $f(0^n) = 1$, then for each min term $d$, we have $\alpha(d) = \sum_{x\preceq d} (-1)^{|d\oplus x|} f(x) = - 1$, which is still non-zero.}, {which contributes 1 to $\mono(f)$}, thus  
	%Let $H(x)$ be the Hamming weight of $x$ where $H(x) = |\{i \in [n] : x_i = 1\}|$.
	%Now take a look at the multilinear polynomial representation of $f = \sum_{S \subseteq [n]} \alpha(S) \prod_{i \in S} x_i$, where $\alpha(S) = \sum_{x : \forall i \in \bar{S}, x_i = 0} (-1)^{|S| - H(x)} f(x)$. Note that for each min term $d$ and $D = \{i:d_i = 1\}$, it has $\alpha(D) = \sum_{d \preceq x} (-1)^{|D| - H(x)} f(x) = f(d) = 1$. So 
	the number of min terms is at most $\mono(f) = r$. Since each 1-input of $f$ is above some min term $d$, the 1-covering number $\cover_1(f)$ has
	\[\cover_1(\fcand) \le \sum_{d: \text{min term}} \cover_1(f_d\circ \wedge) \le r\cdot \max_{d: \text{min term}} \cover_1(f_d\circ \wedge).\]
	Note that $\alt(f_d) \le \alt(f) - 1$, and $f_d$ takes value 1 on its all-0 input, thus $\cover_1(f_d\circ \wedge) = \cover_0(\neg f_d\circ \wedge) \le C_0^{(a-1)}$ (note that the maximum in the definition of $C_0$ is over all $f$ with $f(0^n) = 0$). 
	This implies 
	\begin{align}\label{eq:C1}
		\log C_1^{(a)}\le \ell + \log C_0^{(a-1)}.
	\end{align}
	Note that this inequality holds as long as $f(0^n) = 0$, regardless of the parity of $a$. 
	
	\medskip
	\textbf{Bound 3, from CC.} When $a$ is odd, we have a bound for $C_0^{(a)}$ by Eq.\eqref{eq:C0-odd} and a bound for $C_1^{(a)}$ by Eq.\eqref{eq:C1}. When $a$ is even, we have  two bounds for $C_1^{(a)}$, Eq.\eqref{eq:C1-even} and Eq.\eqref{eq:C1}, but no bound for $C_0^{(a)}$. Note that we can always use \dcc to bound $C_0^{(a)}$: 
	\begin{align*}
		\log \cover_0(\fcand) & = \ncc_0(\fcand) \le \dcc(\fcand) \\
		& \le \log\rank(M_{\fcand})\cdot \log \cover_1(\fcand) = \ell \cdot \log \cover_1(\fcand),
	\end{align*}
	This implies that 
	\begin{align}\label{eq:C0-by-C1}
		\log C_0^{(a)} \le \ell \cdot \log C_1^{(a)}.
	\end{align} 
	Similarly it also holds that 
	$\log C_1^{(a)} \le \ell \cdot \log \C_0^{(a)}$.
	%where the last inequality uses Theorem \ref{thm:CC-N-rank}.
	%\[\log  C_1^{(a)} \le (\log_2 3) \cdot  \ell \cdot \log n \cdot  \log C_1^{(a-1)}.\]
	%By Theorem \ref{thm:CC-N-rank}, we have 
	%\[\dcc(\fcand) \le \log\rank(M_{\fcand})\cdot \log \cover_1(\fcand) \le (\log_2 3) \ell \log n \max_{u: \text{max term}}\cover_1(f_u\circ \wedge)\]
		
	\medskip Now we combine the three bounds and prove the theorem by induction on $a$. In the base case of $a = 0$, the function is constant 0 and thus $C_0^{(0)} = 1$ and $C_1^{(0)} = 0$. For general $a$, we can repeatedly apply Eq.\eqref{eq:C1} and Eq.\eqref{eq:C0-by-C1} to get 
	\[\log C_1^{(a)} \le \sum_{i=1}^a \ell^i = (1+o(1)) \ell^a.\]
	Thus $\dcc(\fcand) \le \ell \cdot \log C_1^{(a)} \le (1+o(1)) \ell^{a+1}$.
	
	If we can tolerate a $\log n$ factor, then the dependence on $a$ can be made slightly better. Assume that $a$ is even, we have 
	\begin{align*}
		\log C_1^{(a)} & \le \ell +  \log C_0^{(a-1)} & \text{(by Eq.\eqref{eq:C1})} \\
		& \le \ell + 3\ell \log n + \log C_0^{(a-2)} & \text{(by Eq.\eqref{eq:C0-odd})} \\
		& \le \ell + 3\ell \log n + \ell \log C_1^{(a-2)}. & \text{(by Eq.\eqref{eq:C0-by-C1})} 
	\end{align*}
	Solving this recursion gives $\log C_1^{(a)} \le O(\ell^{\frac{a}{2}} \log n)$, and thus $\dcc = O(\ell^{\frac{a}{2}+1} \log n)$.	When $a$ is odd, we can use Eq.\eqref{eq:C1} and Eq.\eqref{eq:C0-by-C1} to reduce it to the ``even $a$'' case, resulting a bound $\dcc \le O(\ell^{\frac{a+3}{2}} \log n)$. Putting these two cases together, we get the claimed bound. 
	%The number of 1-rectangled obtained this way is at most
	%\[\left(\binom{n}{\ell} \cdot 3^\ell \cdot r\right)^{\lceil\alt(f)/2\rceil}\]
	%Thus we have the following bound for the communication complexity
	%\[\dcc(\fcand) = O\left(\log \left(\binom{n}{\ell} \cdot 3^\ell \cdot r\right)^{\lceil\alt(f)/2\rceil} \log r \right) = O\left(\alt(f)\log^2 r \log n\right).\]
\end{proof}

%\section{Concluding remarks}

\paragraph*{Acknowledgement}
The authors would like to thank Xin Huang for valuable discussions. S.Z. was supported by Research Grants Council of the Hong Kong S.A.R. (Project no. CUHK419413). Part of the work was done when the S.Z. visited Centre of Quantum Technologies partially under its support.

\bibliographystyle{alpha}
\bibliography{alt}

\end{document}